\documentclass[12pt]{article}

\usepackage[margin=1in]{geometry}
\usepackage{CJKutf8, amsmath, amssymb, amsthm, authblk, complexity, hyperref}

\newtheorem{theorem}{Theorem}
\newtheorem{corollary}{Corollary}
\newtheorem{proposition}{Proposition}

\usepackage[style=trad-unsrt, citestyle=numeric-comp, giveninits=true, doi=false, url=false, eprint=false, isbn=false]{biblatex}
\begin{document}

\title{Neural network representation of tensor network and chiral states}

\begin{CJK}{UTF8}{gbsn}

\author{Yichen Huang (黄溢辰)\thanks{yichuang@mit.edu} \thanks{Present address: Center for Theoretical Physics, Massachusetts Institute of Technology, Cambridge, Massachusetts 02139, USA.}}
\affil{Institute for Quantum Information and Matter, California Institute of Technology, Pasadena, California 91125, USA}

\author{Joel E. Moore\thanks{jemoore@berkeley.edu}}
\affil{Department of Physics, University of California, Berkeley, Berkeley, California 94720, USA}
\affil{Materials Sciences Division, Lawrence Berkeley National Laboratory, Berkeley, California 94720, USA}

\maketitle

\end{CJK}

\begin{abstract}

We study the representational power of Boltzmann machines (a type of neural network) in quantum many-body systems. We prove that any (local) tensor network state has a (local) neural network representation. The construction is almost optimal in the sense that the number of parameters in the neural network representation is almost linear in the number of nonzero parameters in the tensor network representation. Despite the difficulty of representing (gapped) chiral topological states with local tensor networks, we construct a quasi-local neural network representation for a chiral $p$-wave superconductor. These results demonstrate the power of Boltzmann machines.

\end{abstract}

\section{Introduction}

A generic state in quantum many-body systems is classically intractable because the dimension of the Hilbert space grows exponentially with the system size. However, physically relevant states are often non-generic in the sense of having structures, making use of which we may overcome the curse of dimensionality. Traditionally, tensor networks are used to characterize such structures and efficiently represent states in classical simulations \cite{VMC08}. Recently, Carleo and Troyer \cite{CT16} proposed neural networks as an (alternative) ansatz for quantum many-body states. Benchmark calculations suggest that this is a promising approach.

Besides numerical experiments, it is also important to explain the working principle of neural network methods. One step in this direction is to characterize the representational power of neural networks. Here we specialize to Boltzmann machines \cite{AHS85}, and our main contributions are
\begin{itemize}
\item We prove that any (local) tensor network state can be converted into a (local) neural network without significantly increasing the number of parameters.
\item Despite the difficulty of representing (gapped) chiral topological states with local tensor networks \cite{DR15}, we construct a quasi-local neural network representation for a chiral $p$-wave superconductor.
\end{itemize}
The first result states that the representational power of neural networks is at least not weaker than that of tensor networks. The second gives a physically relevant example where neural networks may go beyond tensor networks. In combination, these results provide complementary evidence that neural networks are a promising ansatz.

\section{Boltzmann machines}

We provide a minimum background for those people with no prior knowledge of Boltzmann machines. The goal is to motivate the definition of neural network states, rather than a general-purpose introduction from the perspective of machine learning.

Formally, a Boltzmann machine is a type of stochastic recurrent neural network. In the language of physicists, it is a classical Ising model on a weighted undirected graph. Each vertex (also known as a unit or a neuron) of the graph carries a classical Ising variable $s_j=\pm1$ and a local field $h_j\in\mathbb R$, where $j$ is the index of the vertex. For a reason that will soon be clear, the set of vertices is divided into the disjoint union of two subsets $V$ and $H$ so that $|V|+|H|$ is the total number of units. Vertices in $V$ are called visible units, and those in $H$ are called hidden units. For notational simplicity, we assume that visible units have small indices $1,2,\ldots,|V|$, and hidden units have large indices $|V|+1,|V|+2,\ldots,|V|+|H|$. Each edge of the graph carries a weight $w_{jk}\in\mathbb R$ that describes the interaction between $s_j$ and $s_k$. The energy of a configuration is given by
\begin{equation}
E(\{s_j\})=\sum_{j}h_js_j+\sum_{j,k}w_{jk}s_js_k.
\end{equation}
A restricted Boltzmann machine \cite{Smo86} is a Boltzmann machine on a bipartite graph, i.e., $w_{jk}\neq0$ only if the edge $(j,k)$ connects a visible unit and a hidden unit.

At thermal equilibrium, the configurations follow the Boltzmann distribution. Without loss of generality, we fix the temperature $T=1$. Let
\begin{equation}
Z=\sum_{\{s_j\}\in\{\pm1\}^{\times(|V|+|H|)}}e^{-E(\{s_j\})}
\end{equation}
be the partition function. The probability of each configuration $\{s_j\}$ is given by $e^{-E(\{s_j\})}/Z$. Furthermore, the probability of each configuration $\{s_{j\le|V|}\}$ of visible units is the marginal probability obtained by summing over hidden units:
\begin{equation} \label{bm}
P\{s_{j\le|V|}\}=\frac{1}{Z}\sum_{\{s_{j>|V|}\}\in\{\pm1\}^{\times|H|}}e^{-E(\{s_j\})}.
\end{equation}

The process of training a Boltzmann machine is specified as follows. The input is a (normalized) probability distribution $Q$ over the configurations of the visible units. (Here we assume, for simplicity, that $Q$ is given. In practice, $Q$ may not be explicitly given but we are allowed to sample from $Q$.) The goal is to adjust the weights $w_{jk}$ and local fields $h_j$ such that the probability distribution $P$ (\ref{bm}) best approximates (as measured by the Kullback-Leibler divergence or some other distance function) the desired distribution $Q$. This is a variational minimization problem. We minimize a given objective function of $P$ with respect to the ansatz (\ref{bm}), in which the weights and local fields are variational parameters. Note that Monte Carlo techniques are usually used in training a Boltzmann machine.

An interesting question is whether an arbitrary $Q$ can be exactly represented by a Boltzmann machine. Because of the strict positivity of exponential functions, it is easy to see that the answer is no if the probability of some configuration is zero. Nevertheless,

\begin{theorem} [Le Roux and Bengio \cite{LB08}] \label{rbmc}
Any probability distribution $Q$ can be arbitrarily well approximated by a restricted Boltzmann machine, provided that the number $|H|$ of hidden units is the number of configurations with nonzero probability. In general, $|H|\le2^{|V|}$.
\end{theorem}

This result can be slightly improved \cite{MA11}. Indeed, a simple counting argument suggests that an exponential number of hidden units is necessary in general. As a function from $\{\pm1\}^{\times|V|}$ to $[0,1]$ with one constraint (normalization), the probability distribution $Q$ has $2^{|V|}-1$ degrees of freedom. Therefore, a good approximation of $Q$ requires an exponential number of bits.

\section{Neural network states}

Carleo and Troyer \cite{CT16} developed a minimal extension of Boltzmann machines to the quantum world. In this extension, each vertex still carries a classical spin (bit). This is very different from another extension \cite{AAR+16}, in which each vertex carries a quantum spin (qubit). It should be clear that the latter extension is more quantum. In the former extension, we use classical computers to simulate quantum many-body systems.

Expanded in the computational basis $\{\pm1\}^{\times|V|}$, a quantum state $|\psi\rangle$ of $|V|$ qubits can be viewed as a function from $\{\pm1\}^{\times|V|}$ to $\bar B(0,1)$ with one constraint (normalization), where $\bar B(0,1)$ denotes the closed region $|z|\le1$ in the complex plane. Recall that a probability distribution $Q$ is characterized by the probability of each configuration. In comparison, a state $|\psi\rangle$ is characterized by the probability amplitude of each configuration. This analog between classical probability distributions and quantum states motivates the following ansatz dubbed neural network states.

Consider a graph as before. The visible units correspond to physical qubits, and hidden units are auxiliary degrees of freedom (to be summed over). The local field $h_j\in\mathbb C$ at each vertex and the weight $w_{jk}\in\mathbb C$ carried by each edge are promoted to complex numbers because probability amplitudes are generally complex. An (unnormalized) neural network state based on a Boltzmann machine is given by
\begin{equation} \label{qbm}
|\psi\rangle=\sum_{\{s_1,s_2,\ldots,s_{|V|+|H|}\}\in\{\pm1\}^{\times(|V|+|H|)}}e^{-\sum_{j}h_js_j-\sum_{j,k}w_{jk}s_js_k}|\{s_1,s_2,\ldots,s_{|V|}\}\rangle.
\end{equation}
Note that summing over hidden units is very different from taking partial trace, for the latter usually results in a mixed state. Similarly, a neural network state based on a restricted Boltzmann machine is given by Eq. (\ref{qbm}) on a bipartite graph as described previously.

A very minor modification of the proof of Theorem \ref{rbmc} leads to

\begin{corollary} \label{rbmq}
Any $N$-qubit quantum state $|\psi\rangle$ can be arbitrarily well approximated by a neural network state based on a restricted Boltzmann machine with $w_{jk}\in\mathbb R, h_j\in\mathbb C$, provided that the number $|H|$ of hidden units is the number of configurations with nonzero probability amplitude. In general, $|H|\le2^N$.
\end{corollary}

The formalism of neural network states developed thus far applies to general quantum many-body systems. We now consider the situation that the qubits are arranged on a lattice. The lattice allows us to introduce the notion of locality \cite{Has12} (with respect to the shortest path metric), which underlies almost all successful methods for simulating quantum lattice systems. Hence, it is desirable to incorporate locality into the neural network. To this end, we define a position for each unit. Let each site of the lattice carry a visible unit and some hidden units. We require that $w_{jk}\neq0$ only if units $j$ and $k$ are close to each other. As an example, Deng et al. \cite{DLD16} showed that the ground state of the toric code in two dimensions can be exactly represented as a local neural network state based on a restricted Boltzmann machine.

Using variational quantum Monte Carlo techniques (``quantum Monte Carlo'' is not a quantum algorithm; rather, it is just a classical Monte Carlo algorithm applied to quantum systems), Carleo and Troyer \cite{CT16} performed practical calculations for quantum lattice systems in one and two spatial dimensions. For the models they studied, they observed that variational approaches based on neural networks are at least as good as those based on tensor networks. It seems worthwhile to devote more study to the practical performance of neural network states.

\section{Tensor network states}

We provide a very brief introduction to tensor network states following the presentation in Ref. \cite{GHLS15}, Subsection 6.3 or Ref. \cite{Hua15}, Subsection 2.3. Then, it will be clear almost immediately that any (local) tensor network state has a (local) neural network representation.

An $l$-dimensional tensor $T$ is a multivariable function $T:\{1,2,\ldots,d_1\}\times\{1,2,\ldots,d_2\}\times\cdots\times\{1,2,\ldots,d_l\}\to\mathbb C$, and $D=\max_jd_j$ is called the bond dimension. It is easy to see that $T$ can be reshaped to an $l'$-dimensional ($l'=\sum_j\lceil\log_2 d_j\rceil$ with $\lceil\cdot\rceil$ the ceiling function) tensor $T'$ with bond dimension $2$ by representing every input variable in binary. Moreover, there is a straightforward way to identify $T'$ with an $l'$-qubit unnormalized state expanded in the computational basis:
\begin{equation} \label{N}
|\psi\rangle=\sum_{\{s_j\}\in\{\pm1\}^{\times l'}}T'\left(\frac{s_1+3}{2},\frac{s_2+3}{2},\dots,\frac{s_n+3}{2}\right)|\{s_j\}\rangle.
\end{equation}
Thus, a neural network state with $|V|$ visible units is a $|V|$-dimensional tensor with bond dimension $2$. Indeed, each visible unit corresponds to an input variable to the tensor. As a restatement of Corollary \ref{rbmq},
\begin{corollary} \label{rbmt}
Any $l$-dimensional tensor $T$ can be arbitrarily well approximated by a restricted Boltzmann machine with $w_{jk}\in\mathbb R, h_j\in\mathbb C$, provided that the number $|H|$ of hidden units is the number of nonzero elements in $T$. In general, $|H|\le\prod_{j=1}^ld_j$.
\end{corollary}

Informally, tensor contraction is defined as follows. Suppose we have two three-dimensional tensors $T_1(j_1,j_2,j_3),T_2(j'_1,j'_2,j'_3)$ of shapes $d_1\times d_2\times d_3,d'_1\times d_2\times d'_3$, respectively. Their contraction on the middle indices is a four-dimensional tensor
\begin{equation}
\sum_{j=1}^{d_2}T_1(j_1,j,j_3)T_2(j'_1,j,j'_3).
\end{equation}

Similarly, suppose we have two Boltzmann machines with units $\{s_j\}=V\cup H,\{s'_j\}=V'\cup H'$, weights $w_{jk},w'_{jk}$, and local fields $h_j,h'_j$, respectively. Their contraction on the first $c$ visible units is defined as identifying $s_j$ with $s'_j$ and then summing over $s_j$ for $j=1,2,\ldots,c$:
\begin{equation}
\sum_{\{s_1,s_2,\ldots,s_{|V|+|H|}\}\in\{\pm1\}^{\times(|V|+|H|)},\{s'_1,s'_2,\ldots,s'_{|V'|+|H'|}\}\in\{\pm1\}^{\times(|V'|+|H'|)}~\textnormal{s.t.}~s_1=s'_1,s_2=s'_2,\ldots,s_c=s'_c}\cdots,
\end{equation}
where $\cdots$ is given by
\begin{equation}
e^{-\sum_{j}h_js_j-\sum_{j,k}w_{jk}s_js_k-\sum_{j}h'_js'_j-\sum_{j,k}w'_{jk}s'_js'_k}|\{s_{c+1},s_{c+2},\ldots,s_{|V|}\}\rangle\otimes|\{s'_{c+1},s'_{c+2},\ldots,s'_{|V'|}\}\rangle.
\end{equation}
Note that only visible units are allowed to be contracted on, and such visible units become hidden after the contraction. Thus, the Boltzmann machine after contraction has $|V|+|V'|-2c$ visible and $|H|+|H'|+c$ hidden units.

A multi-qubit state $|\psi\rangle$ is a tensor network state if its tensor representation $T'$ (\ref{N}) can be obtained by contracting a network of $C$-dimensional tensors, where $C$ is a small absolute constant. The bond dimension $D$ of a tensor network is defined as the maximum bond dimension of the constituting tensors in the network.

\begin{theorem} \label{main}
Any tensor network state $|\psi\rangle$ can be arbitrarily well approximated by a neural network state based on a Boltzmann machine with $w_{jk}\in\mathbb R, h_j\in\mathbb C$, provided that the number $|H|$ of hidden units is sufficiently large. Let $\emptyset(T_j)$ be the number of nonzero elements in a constituting tensor $T_j$, and $\#(T_j)$ be the number of elements. It suffices that
\begin{equation} \label{9}
|H|=\sum_{j=1}^n(\emptyset(T_j)+\log_2\#(T_j)+O(1)),
\end{equation}
where $n$ is the number of constituting tensors. Furthermore, the number of parameters in the neural network representation is upper bounded by
\begin{equation} \label{10}
\sum_j\emptyset(T_j)\log_2\#(T_j)+\textnormal{less significant terms}.
\end{equation}
The neural network is local (translationally invariant) if the tensor network is local (translationally invariant).
\end{theorem}

\begin{proof}
We first represent each $T_j$ with a restricted Boltzmann machine as in Corollary \ref{rbmt}, and then contract the restricted Boltzmann machines in the same way as $T_j$'s are contracted. Note that the Boltzmann machine after contraction is generically not restricted. The first term on the right-hand side in Eq. (\ref{9}) is the total number of hidden units in all constituting restricted Boltzmann machines, and the other terms are responsible for the production of hidden units in the contraction process. Equation (\ref{10}) is the total number of parameters in all constituting restricted Boltzmann machines, and the contraction process does not introduce any new parameters. It is obvious that the construction is locality and translational-invariance preserving.
\end{proof}

This result is almost optimal in the sense that the number of parameters in the neural network representation is at most a logarithmic (in the bond dimension) multiple of the number of nonzero parameters in the tensor network representation. In cases that the tensors are sparse (possibly due to the presence of symmetries), this conversion to neural network states has the desirable property of automatically compressing the representation.

As an example, we specialize Theorem \ref{main} to matrix product states \cite{FNW92, PVWC07}.

\begin{corollary} \label{mps}
To leading order, any $N$-qubit matrix product state with bond dimension $D$ has a neural network representation with $2ND^2$ hidden units and $4ND^2\log_2 D$ parameters.
\end{corollary}

\section{Chiral topological states}

It seems difficult to obtain a local tensor network representation for gapped chiral topological states. Early attempts did not impose locality \cite{GVW10} or just target at expectation values of local observables rather than the wave function \cite{BC11}. Recent progress \cite{WTSC13, DR15} shows that local tensor networks can describe chiral topological states, but the examples there are all gapless. Indeed, Dubail and Read \cite{DR15} even proved a no-go theorem, which roughly states that for any chiral local free-fermion tensor network state, any local parent Hamiltonian is gapless. Here we construct a Boltzmann machine that approximates the unique ground state of a (gapped) chiral $p$-wave superconductor. The approximation error is inverse polynomial in the system size, and the neural network is quasi-local in the sense that the maximum distance of connections between units is logarithmic in the system size. This example explicitly demonstrates the power of Boltzmann machines.

From now on, we consider fermionic systems, for which it is necessary and there are multiple ways to account for the exchange statistics of fermions. We take a straightforward approach: Each vertex carries a Grassmann variable $\xi_j$ rather than an Ising variable $s_j$, and the sum over Ising variables in hidden units is replaced by the Grassmann integral. We work in the second quantization formalism. Let $c_j,c_j^\dag$ for $1\le j\le|V|$ be the fermionic annihilation and creation operators, and $|0\rangle$ be the vacuum state with no fermions. We identify $c_j^\dag$ with $\xi_j$ so that $\Xi|0\rangle$ represents a fermionic state, where $\Xi$ is an arbitrary Grassmann variable in the algebra generated by $\xi_1,\xi_2,\ldots,\xi_{|V|}$. As an analog of Eq. (\ref{qbm}),
\begin{equation} \label{fnn}
|\psi\rangle=\left(\int e^{-\sum_{j,k}w_{jk}\xi_j\xi_k}\prod_{l=1}^{|H|}\mathrm d\xi_{|V|+l}\right)|0\rangle
\end{equation}
is an (unnormalized Gaussian) fermionic neural network state. Note that we have set $h_j=0$. This class of states appeared previously in Ref. \cite{DR15}.

One of the simplest examples of a chiral topological phase is the $p+ip$ superconductor \cite{RG00}. For concreteness, we consider the lattice model in Ref. \cite{BN15}. Let $c_{\vec x},c_{\vec x}^\dag$ be the fermionic annihilation and creation operators at site $\vec x\in\mathbb Z^2$ on a two-dimensional square lattice. Let $\vec i=(1,0)$ and $\vec j=(0,1)$ be the unit vectors in the $x$ and $y$ axes, respectively. The Hamiltonian is
\begin{equation} \label{model}
H=\sum_{\vec x\in\mathbb Z^2}c_{\vec x+\vec i}^\dag c_{\vec x}+c_{\vec x+\vec j}^\dag c_{\vec x}+c_{\vec x+\vec i}^\dag c_{\vec x}^\dag+ic_{\vec x+\vec j}^\dag c_{\vec x}^\dag+{\rm H.c.}-2\mu\sum_{\vec x\in\mathbb Z^2}c_{\vec x}^\dag c_{\vec x},\quad\mu\in\mathbb R.
\end{equation}
Let $\vec k=(k_x,k_y)$ be the lattice momentum, and $\int_{\textnormal{BZ}}\mathrm d\vec k$ be the integral over the Brillouin zone $(-\pi,\pi]^{\times2}$. To slightly simplify the presentation, we will be sloppy about unimportant overall prefactors in the calculations below. The Fourier transform
\begin{equation} \label{ft}
c_{\vec x}^\dag=\int_{\textnormal{BZ}}e^{-i\vec k\cdot\vec x}c_{\vec k}^\dag\,\mathrm d\vec k,\quad c_{\vec x}=\int_{\textnormal{BZ}}e^{i\vec k\cdot\vec x}c_{\vec k}\,\mathrm d\vec k
\end{equation}
leads to
\begin{equation}
H=\int_{\textnormal{BZ}}H_{\vec k}\,\mathrm d\vec k,\quad H_{\vec k}=\Delta_{\vec k}c_{\vec k}^\dag c_{-\vec k}^\dag+\Delta_{\vec k}^*c_{-\vec k}c_{\vec k}+2M_{\vec k}c_{\vec k}^\dag c_{\vec k}
\end{equation}
in the momentum space, where
\begin{equation}
\Delta_{\vec k}=\sin k_y-i\sin k_x,\quad M_{\vec k}=\cos k_x+\cos k_y-\mu.
\end{equation}
The quasiparticle spectrum is given by
\begin{equation}
E_{\vec k}=2\sqrt{|\Delta_{\vec k}|^2+M_{\vec k}^2}=2\sqrt{\sin^2k_x+\sin^2k_y+(\cos k_x+\cos k_y-\mu)^2}.
\end{equation}
Hence, $H$ is gapless for $\mu=0,\pm2$ and gapped otherwise. The unique ground state of $H$ is
\begin{equation} \label{ff}
|\psi\rangle\propto e^{\frac12\int_{\textnormal{BZ}}\frac{v_{\vec k}}{u_{\vec k}}c_{\vec k}^\dag c_{-\vec k}^\dag\,\mathrm d\vec k}|0\rangle,
\end{equation}
where $|0\rangle$ is the vacuum state, and
\begin{equation} \label{uv}
u_{\vec k}=-\sqrt{|\Delta_{\vec k}|^2+M_{\vec k}^2}-M_{\vec k},\quad v_{\vec k}=\Delta_{\vec k}
\end{equation}
so that $u_{\vec k}=u_{-\vec k}$ and $v_{\vec k}=-v_{-\vec k}$. It is not difficult to see that the model (\ref{model}) is a trivial superconductor for $|\mu|>2$. It is topological superconductors with opposite chirality for $-2<\mu<0$ and $0<\mu<2$, respectively.

Indeed, a state of the form (\ref{ff}) has an exact local neural network representation (\ref{fnn}) \cite{DR15} if $u_{\vec k},v_{\vec k}$ are constant-degree trigonometric polynomials, i.e., polynomials in $e^{\pm ik_x},e^{\pm ik_y}$. In particular, we have each site $\vec x$ carry a visible and a hidden unit. To simplify the notation, the Grassmann variable in the visible unit is identified with $c_{\vec x}^\dag$, and that in the hidden unit is denoted by $\xi_{\vec x}$. The Fourier transform of Grassmann variables is defined as
\begin{equation}
\xi_{\vec x}=\int_{\textnormal BZ}e^{-i\vec k\cdot\vec x}\xi_{\vec k}\,\mathrm d\vec k.
\end{equation}
In the momentum space, it is easy to see that the state (\ref{ff}) can be represented as
\begin{equation} \label{tp}
|\psi\rangle\propto\left(\int[\mathrm d\xi_{\vec k}]e^{\int_{\textnormal{BZ}}(v_{\vec k}\xi_{-\vec k}c_{\vec k}^\dag-\frac12\xi_{-\vec k}v_{-\vec k}u_{\vec k}\xi_{\vec k})\,\mathrm d\vec k}\right)|0\rangle,
\end{equation}
where $\int[\mathrm d\xi_{\vec k}]$ denotes the integral over all Grassmann variables in the momentum space with a proper measure. Transforming to the real space, $[\mathrm d\xi_{\vec k}]$ becomes $\prod_{\vec x\in\mathbb Z^2}\mathrm d\xi_{\vec x}$, and the exponent in parentheses is local because $u_{\vec k},v_{\vec k}$ are trigonometric polynomials (the maximum distance of connections between units is proportional to the degree of the trigonometric polynomials). Thus, Eq. (\ref{tp}) reduces to Eq. (\ref{fnn}) with additional properties: (i) the neural network representation is translationally invariant; (ii) there are no connections between visible units.

In the remainder of this section, standard asymptotic notations are used extensively. Let $f,g:\mathbb R^+\to\mathbb R^+$ be two functions. One writes $f(x)=O(g(x))$ if and only if there exist constants $M,x_0>0$ such that $f(x)\le Mg(x)$ for all $x>x_0$; $f(x)=\Omega(g(x))$ if and only if there exist constants $M,x_0>0$ such that $f(x)\ge Mg(x)$ for all $x>x_0$; $f(x)=o(g(x))$ if and only if for any constant $M>0$ there exists a constant $x_0>0$ such that $f(x)<Mg(x)$ for all $x>x_0$.

Generically, exact representations are too much to ask for; hence, approximate representations are acceptable. Furthermore, we usually have to work with a finite system size in order to rigorously quantify the approximation error and the succinctness of the representation. We now justify these statements with an example. Matrix product states are an excellent ansatz for ground states in one-dimensional gapped systems \cite{Has07}. A generic gapped ground state in a chain of $N$ spins (i) cannot be exactly written as a matrix product state with bond dimension $e^{o(N)}$; (ii) is not expected to be approximated (in the sense of $99\%$ fidelity) by a matrix product state with bond dimension $O(1)$ because errors may accumulate while we truncate the Schmidt coefficients across every cut \cite{VC06}; (iii) can be approximated by a matrix product state with bond dimension $N^{o(1)}$ \cite{AKLV13, Hua14}.

The goal is to construct an as-local-as-possible neural network representation (\ref{fnn}) for the ground state of the model (\ref{model}).

\begin{proposition}
Let $|\psi\rangle$ be the ground state of the model (\ref{model}) on a (finite) square lattice of size $L\times L$ with periodic boundary conditions. There exists a neural network state $|\phi\rangle$ such that the fidelity $|\langle\phi|\psi\rangle|\ge1-1/\poly(L)$ and that the maximum distance of connections between units is $O(\epsilon^{-1}\log L)$, where $\epsilon$ is the energy gap.
\end{proposition}

\begin{proof}
For concreteness, we consider $0<\mu<2$. It should be clear that the same result holds for other values of $\mu$ provided that the system is gapped. The main idea is to approximate $u_{\vec k}$ (\ref{uv}) with a trigonometric polynomial. The approximation error is exponentially small in the degree of the polynomial because $u_{\vec k}$ is a real analytic function of $\vec k$. The locality of the neural network is due to the smallness of the degree.

We now provide the details of the construction. Assume $L$ is odd. This slightly simplifies the analysis, but it should be clear that a minor modification of the analysis leads to the same result for even $L$. We abuse the notation by letting $\textnormal{BZ}$ denote the set of viable points $\{0,\pm2\pi/L,\ldots,\pm(L-1)\pi/L\}^{\times2}$ for lattice momenta in the Brillouin zone, and
\begin{equation}
\textnormal{hBZ}=\{\vec k\in\textnormal{BZ}|k_x>0\lor(k_x=0\land k_y>0)\}
\end{equation}
be the ``right half'' of $\textnormal{BZ}$. The Fourier transform (\ref{ft}) becomes discrete:
\begin{equation}
c_{\vec x}^\dag=\frac1{\sqrt L}\sum_{\vec k\in\mathrm{BZ}}e^{-i\vec k\cdot\vec x}c_{\vec k}^\dag,\quad c_{\vec x}=\frac1{\sqrt L}\sum_{\vec k\in\mathrm{BZ}}e^{i\vec k\cdot\vec x}c_{\vec k},
\end{equation}
and the Hamiltonian in the momentum space is given by
\begin{equation}
H=2\sum_{\vec k\in\textnormal{hBZ}}(M_{\vec k}c_{\vec k}^\dag c_{\vec k}+M_{\vec k}c_{-\vec k}^\dag c_{-\vec k}+\Delta_{\vec k}c_{\vec k}^\dag c_{-\vec k}^\dag+\Delta_{\vec k}^*c_{-\vec k}c_{\vec k})+2M_{\vec k=\vec{0}}c_{\vec k=\vec{0}}^\dag c_{\vec k=\vec{0}}.
\end{equation}
The normalized ground state is
\begin{equation}
|\psi\rangle=\bigotimes_{\vec k\in\textnormal{hBZ}}|\psi_{\vec k}\rangle,\quad|\psi_{\vec k}\rangle\propto e^{v_{\vec k}c_{\vec k}^\dag c_{-\vec k}^\dag/u_{\vec k}}|0\rangle.
\end{equation}
This equation should be understood as the $\vec k=\vec 0$ mode of $|\psi\rangle$ being vacant. Consider the normalized state
\begin{equation}
|\phi^{(m)}\rangle=\bigotimes_{\vec k\in\textnormal{hBZ}}|\phi^{(m)}_{\vec k}\rangle,\quad|\phi^{(m)}_{\vec k}\rangle\propto e^{v_{\vec k}c_{\vec k}^\dag c_{-\vec k}^\dag/u^{(m)}_{\vec k}}|0\rangle,
\end{equation}
where $u^{(m)}_{\vec k}$ is a degree-$m$ trigonometric polynomial obtained by expanding $u_{\vec k}$ in Fourier series and computing the partial sum up to order $m$. Similar to Eq. (\ref{tp}), $|\phi^{(m)}\rangle$ has an exact neural network state representation (\ref{fnn}) such that the maximum distance of connections between units is $O(m)$.

As $u_{\vec k}$ is a real analytic function of $\vec k$, its Fourier coefficients of order $m$ decay exponentially as $e^{-m/\xi}$ for some constant $\xi>0$. The decay rate $\xi$ is a function of $\mu$ and can be solved analytically; see, e.g., Ref. \cite{Tre96}, Chapter 2. In the regime the energy gap $\epsilon$ is small, we obtain $\xi=O(1/\epsilon)$. Therefore,
\begin{equation} \label{25}
|u^{(m)}_{\vec k}-u_{\vec k}|=e^{-\Omega(\epsilon m)},\quad\forall\vec k\in\textnormal{hBZ}.
\end{equation}
Furthermore, the absolute values of $u_{\vec k},v_{\vec k}$ are bounded away from $0$:
\begin{equation} \label{26}
|u_{\vec k}|\ge\Omega(L^{-2}),\quad|v_{\vec k}|\ge\Omega(L^{-1}),\quad\forall\vec k\in\textnormal{hBZ}.
\end{equation}
Equations (\ref{25}), (\ref{26}) imply
\begin{equation}
f_{\vec k}:=|\langle\phi^{(m)}_{\vec k}|\psi_{\vec k}\rangle|\ge1-1/\poly(L),\quad\forall\vec k\in\textnormal{hBZ}
\end{equation}
for $m=O(\epsilon^{-1}\log L)$ with a sufficiently large constant prefactor hidden in the big-O notation. As $|\psi\rangle,|\phi^{(m)}\rangle$ are product states in the momentum space, the fidelity is given by
\begin{equation}
|\langle\phi^{(m)}|\psi\rangle|=\prod_{\vec k\in\mathrm{hBZ}}f_{\vec k}\ge(1-1/\poly(L))^{O(L^2)}=1-\frac{1}{\poly(L)}.
\end{equation}
We complete the proof by letting $|\phi\rangle=|\phi^{(m)}\rangle$.
\end{proof}

It is not difficult to see that $|\psi\rangle$ is well approximated by a thermal state at inverse temperature $O(\log L)$. Since the thermal state has a projected entangled pair approximation with quasi-polynomial bond dimension $e^{O(\log^2L)}$ \cite{Has06, KGK+14, MSVC15}, there exists a projected entangled pair state $|\varphi\rangle$ with bond dimension $e^{O(\log^2L)}$ such that $|\langle\varphi|\psi\rangle|\ge1-1/\poly(L)$. It is an open problem to improve the bond dimension of $|\varphi\rangle$ to $\poly(L)$.

\section*{Acknowledgments and notes}

The authors would like to thank Xie Chen for collaboration in the early stages of this project. Y.H. acknowledges funding provided by the Institute for Quantum Information and Matter, an NSF Physics Frontiers Center (NSF Grant PHY-1733907) with support of the Gordon and Betty Moore Foundation (GBMF-2644). J.E.M. is supported by NSF DMR-1507141, DMR-1918065 and a Simons Investigatorship.

Part of this work was presented on November 27, 2014 (Thanksgiving day!) at the Perimeter Institute for Theoretical Physics. Very recently, we became aware of some related papers \cite{CCX+17, DLD17, GD17}, which studied the relationship between neural and tensor network states using different methods. In particular, Theorem \ref{main} and Corollary \ref{mps} are stronger than Theorem 3 in Ref. \cite{GD17}. After the present work had been on arXiv, some other related papers \cite{NDYI17, RS18, Cla18, GPA+18, KLB17, CNI18, LGD19, LPZZ21} appeared.

\printbibliography

\end{document}